\documentclass{article}
\usepackage{nips15submit_e,times}

\usepackage{amsmath,amsfonts,graphpap,amscd,mathrsfs,graphicx,lscape}
\usepackage{epsfig,amssymb,amstext,xspace}

\usepackage{color}              
\usepackage{epsfig}

\usepackage[suppress]{color-edits}
\addauthor{et}{green}
\addauthor{vs}{red}

\usepackage{amsthm}

\usepackage[ruled,vlined]{algorithm2e}
\renewcommand{\algorithmcfname}{ALGORITHM}
\SetAlFnt{\small}
\SetAlCapFnt{\small}
\SetAlCapNameFnt{\small}
\SetAlCapHSkip{0pt}



\newtheorem{theorem}{Theorem}
\newtheorem{lemma}[theorem]{Lemma}

\newtheorem{corollary}[theorem]{Corollary}
\newtheorem{defn}[theorem]{Definition}

\def\squareforqed{\hbox{\rule{2.5mm}{2.5mm}}}

\def\QED{\ifmmode\squareforqed 
  \else{\nobreak\hfil   
    \penalty50                 
    \hskip1em                  
    \null                      
    \nobreak                   
    \hfil                      
    \squareforqed              
    \parfillskip=0pt           
    \finalhyphendemerits=0     
    \endgraf}                  
  \fi}

\def\blksquare{\rule{2mm}{2mm}}
\def\qedsymbol{\blksquare}
\newcommand{\bg}[1]{\medskip\noindent{\bf #1}}
\newcommand{\ed}{{\hfill\qedsymbol}\medskip}

\newenvironment{remark}{\bg{Remark. }}{\ed}

\newcommand{\AG}{\ensuremath{\texttt{AG}}}

\newcommand{\comment}[1]{}
 {}

\newenvironment{rtheorem}[3][]{
\bigskip
\noindent \ifthenelse{\equal{#1}{}}{\bf #2 #3}{\bf #2 #3 (#1)}
\begin{it}
}{\end{it}}

\newcommand{\E}{\mathbb{E}}

\newcommand{\rev}{R}

\newcommand{\F}{\ensuremath{\mathcal F}}

\newcommand{\T}{\ensuremath{\mathcal T}}

\newcommand{\junk}[1]{}

\newlength{\tmp} \newlength{\lpsx} \newlength{\lpsy} \newlength{\upsx} \newlength{\upsy}

\newcommand{\poa}{\text{\textsc{PoA}} }

\newcommand{\opt}{\text{\textsc{Opt}} }

\newcommand{\pota}{\text{\textsc{PoTA}} }

\newcommand{\M}{\ensuremath{\mathcal M}}

\newcommand{\X}{\ensuremath{\mathcal X}}

\newcommand{\V}{\ensuremath{\mathcal V}}

\newcommand{\val}{\ensuremath{v}}

\newcommand{\vr}{\ensuremath{\mathbf{v}}}

\newcommand{\wb}{\ensuremath{\mathbf{w}}}

\newcommand{\A}{\ensuremath{\mathcal A}}

\newcommand{\al}{\ensuremath{\mathbf a}}

\newcommand{\sr}{\ensuremath{\mathbf s}}

\newcommand{\CCE}{\text{\textsc{CCE}}}

\newcommand{\BCCE}{\text{\textsc{Bayes-CCE}}}

\newcommand{\ex}{\texttt{AG}}

\newcommand{\specialcell}[2][c]{%
  \begin{tabular}[#1]{@{}c@{}}#2\end{tabular}}

\nipsfinalcopy 

\begin{document}

\title{No-Regret Learning in Bayesian Games}

\author{Jason Hartline\\Northwestern University\\Evanston, IL\\{\tt hartline@northwestern.edu}
\And
Vasilis Syrgkanis\\Microsoft Research\\New York, NY\\{\tt vasy@microsoft.com}
\And
\'{E}va Tardos\\Cornell University\\Ithaca, NY\\{\tt eva@cs.cornell.edu}
}
\date{}

\maketitle

\begin{abstract}
Recent price-of-anarchy analyses of games of complete information suggest that coarse correlated equilibria, which characterize outcomes resulting from no-regret learning dynamics, have near-optimal welfare. This work provides two main technical results that lift this conclusion to games of incomplete information, a.k.a., Bayesian games. First, near-optimal welfare in Bayesian games follows directly from the smoothness-based proof of near-optimal welfare in the same game when the private information is public.  Second, no-regret learning dynamics converge to Bayesian coarse correlated equilibrium in these incomplete information games. These results are enabled by interpretation of a Bayesian game as a stochastic game of complete information.

\end{abstract}
\section{Introduction}

A recent confluence of results from game theory and learning theory
gives a simple explanation for why good outcomes in large families of
strategically-complex games can be expected.  The advance comes from
(a) a relaxation the classical notion of equilibrium in games to one
that corresponds to the outcome attained when players' behavior
ensures asymptotic {\em no-regret}, e.g., via standard online learning
algorithms such as {\em weighted majority}, and (b) an extension
theorem that shows that the standard approach for bounding the quality
of classical equilibria automatically implies the same bounds on the
quality of no-regret equilibria.  This paper generalizes these results
from static games to Bayesian games, for example, auctions.

Our motivation for considering learning outcomes in Bayesian games is
the following.  Many important games model
repeated interactions between an uncertain set of
participants. Sponsored search, and more generally, online ad-auction
market places, are important examples of such games. Platforms are
running millions of auctions, with each individual auction slightly
different and of only very small value, but such market places have
high enough volume to be the financial basis of large industries. This
online auction environment is best modeled by a repeated Bayesian
game: the auction game is repeated over time, with the set of
participants slightly different each time, depending on many factors
from budgets of the players to subtle differences in the
opportunities.

A canonical example to which our methods apply is a single-item
first-price auction with players' values for the item drawn from a
product distribution. In such an auction, players simultaneously
submit sealed bids and the player with the highest bid wins and pays
her bid.  The utility of the winner is her value minus her bid; the
utilities of the losers are zero.  When the values are drawn from
non-identical continuous distributions the Bayes-Nash equilibrium is
given by a differential equation that is not generally analytically
tractable, cf.~\cite{Kaplan2012} (and generalizations of this model,
computationally hard, see \cite{Cai2014}).
Again, though their Bayes-Nash equilibria are complex, we show that
good outcomes can be expected in these kinds of auctions.

Our approach to proving that good equilibria can be expected in
repeated Bayesian games is to extend an analogous result for static
games,\footnote{In the standard terms of the game theory literature,
we extend results for learning in games of complete information to
games of incomplete information.} i.e., the setting where the same
game with the same payoffs and the same players is repeated.  Nash
equilibrium is the classical model of equilibrium for each stage of
the static game.  In such an equilibrium the strategies of players may
be randomized; however, the randomizations of the players are
independent.  To measure the quality of outcomes in games
Koutsoupias and Papadimitriou~\cite{Koutsoupias1999} introduced the \emph{price of anarchy}, the 
ratio of the quality of the worst Nash equilibrium over a socially
optimal solution.  Price of anarchy results have been shown for
large families of games, with a focus on those relevant for computer
networks.  Roughgarden~\cite{Roughgarden2009} identified the canonical approach
for bounding the price of anarchy of a game as showing that it satisfies a
natural {\em smoothness} condition.

There are two fundamental flaws with Nash equilibrium as a description
of strategic behavior.  First, computing a Nash equilibrium can be
PPAD hard and, thus, neither should efficient algorithms for computing
a Nash equilibrium be expected nor should any dynamics (of players
with bounded computational capabilities) converge to a Nash
equilibrium.  Second, natural behavior tends to introduce correlations
in strategies and therefore does not converge to Nash equilibrium even
in the limit.  Both of these issues can be resolved for large families
of games.  First, there are relaxations of Nash equilibrium which
allow for correlation in the players' strategies.  Of these, this
paper will focus on {\em coarse correlated equilibrium} which requires
the expected payoff of a player for the correlated strategy be no
worse than the expected payoff of any action at the player's disposal.
Second, it was proven by Blum et al.~\cite{Blum2008} that the
(asymptotic) no-regret property of many online learning algorithms
implies convergence to the set of coarse correlated
equilibria.\footnote{This result is a generalization of one of Foster
and Vohra~\cite{Foster1998}.}

Blum et al.~\cite{Blum2008} extended the definition of the price of anarchy 
to outcomes obtained when each player follows a no-regret learning
algorithm.\footnote{They referred to this price of anarchy for
no-regret learners 
as
the {\em price of total anarchy}.}  As coarse correlated equilibrium
generalize Nash equilibrium it could be that the worst case
equilibrium under the former is worse than the
latter.  Roughgarden~\cite{Roughgarden2009}, however, observed that there is
often no degradation; specifically, the very same smoothness
property that he identified as implying good welfare in Nash equilibrium
also proves good welfare of coarse correlated equilibrium
(equivalently: for outcomes from no-regret learners).  Thus, for a
large family of static games, we can expect strategic behavior to lead
to good outcomes.

\vsedit{This paper extends this theory to Bayesian games. Our
contribution is two-fold: (i) We show an analog of the convergence
of no-regret learning to coarse correlated equilibria in Bayesian
games, which is of interest independently of our price of anarchy
analysis; and (ii) we show that the coarse correlated equilibria of the Bayesian 
version of any smooth static game have good welfare.
Combining these results, we conclude that no-regret learning in smooth
Bayesian games achieves good welfare.}

These results are obtained as follows.  It is possible to view a
Bayesian game as a stochastic game, i.e., where the payoff structure
is fixed but there is a random action on the part of Nature.  This
viewpoint applied to the above auction example considers a population
of bidders associated for each player and, in each stage, Nature
uniformly at random selects one bidder from each population to
participate in the auction.  We re-interpret and strengthen a result
of Syrgkanis and Tardos~\cite{Syrgkanis2013} by showing that the
smoothness property of the static game (for any fixed profile of
bidder values) implies smoothness of this stochastic game.  From the
perspective of coarse correlated equilibrium, there is no difference
between a stochastic game and the non-stochastic game with each random
variable replaced with its expected value.  Thus, the smoothness
framework of Roughgarden~\cite{Roughgarden2009} extends this result to
imply that the coarse correlated equilibria of the stochastic game are
good.  To show that we can expect good outcomes in Bayesian games, it
suffices to show that no-regret learning converges to the coarse
correlated equilibrium of this stochastic game.  Importantly, when we
consider learning algorithms there is a distinction between the
stochastic game where players' payoffs are random variables and the
non-stochastic game where players' payoffs are the expectation of
these variables.  Our analysis addressed this distinction and, in
particular, shows that, in the stochastic game on populations,
no-regret learning converges almost surely to the set of coarse
correlated equilibrium. \vsedit{This result implies that the average
welfare of no-regret dynamics will be good, almost surely, and not
only in expectation over the random draws of Nature.} 

\section{Preliminaries}


This section describes a general game theoretic environment which
includes auctions and resource allocation mechanisms.  For this
general environment we review the results from the literature for
analyzing the social welfare that arises from no-regret learning
dynamics in repeated game play.  The subsequent sections of the paper
will generalize this model and these results to Bayesian games,
a.k.a., games of incomplete information.

\paragraph{General Game Form.} 
A general game $\M$ is specified by a mapping from a profile
$a\in \A\equiv\A_1\times\cdots\times\A_n$ of allowable actions of
players to an outcome.  Behavior in a game may result in (possibly
correlated) randomized actions
$\al \in \Delta(\A)$.\footnote{Bold-face symbols denote random
variables.}  Player $i$'s utility in this game is determined by a
profile of individual values
$\val \in \V\equiv\V_1\times\cdots\times \V_n$ and the (implicit)
outcome of the game; it is denoted $U_i(\al; v_i)
= \E_{a \sim \al}\left[U_i(a; v_i)\right]$. In games with a
social planner or principal who does not take an action in the game,
the utility of the principal is $\rev(\al)
=\E_{a \sim \al}\left[\rev(a)\right]$.  In many games of
interest, such as auctions or allocation mechanisms, the utility of
the principal is the revenue from payments from the players.  We will
use the term {\em mechanism} and {\em game} interchangeably.

In a {\em static game} the payoffs of the players (given by $\val$) are
fixed.  Subsequent sections will consider {\em Bayesian games} in the
independent private value model, i.e., where player $i$'s value $v_i$
is drawn independently from the other players' values and is known
only privately to player $i$.  Classical game theory assumes {\em
complete information} for static games, i.e., that $\val$ is known, and
{\em incomplete information} in Bayesian games, i.e., that the
distribution over $\V$ is known.  For our study of learning in games
no assumptions of knowledge are made; however, to connect to the
classical literature we will use its terminology of complete and
incomplete information to refer to static and Bayesian games,
respectively.

\paragraph{Social Welfare.} 
We will be interested in analyzing the quality of the outcome of the
game as defined by the social welfare, which is the sum of the
utilities of the players and the principal.  We will denote by
$SW(\al;\val)=\sum_{i\in [n]}U_i(\al; v_i)+\rev(\al)$
the expected social welfare of mechanism $\M$ under a randomized
action profile $\al$.  For any valuation profile $\val \in \V$ we will
denote the optimal social welfare, i.e, the maximum over outcomes of the
game of the sum of utilities, by $\opt(\val)$.

\paragraph{No-regret Learning and Coarse Correlated Equilibria.}
For complete information games, i.e., fixed valuation profile $\val$,
Blum et al.~\cite{Blum2008} analyzed repeated play of players using
no-regret learning algorithms, and showed that this play converges to
a relaxation of Nash equilibrium, namely, coarse correlated equilibrium.
\begin{defn}[no regret]
A player achieves {\em no regret} in a sequence of play $a^1,\ldots,a^T$ if his regret
against any fixed strategy $a_i'$ vanishes to zero:
\begin{equation}
\textstyle{\lim_{T\rightarrow \infty}\frac{1}{T}\sum_{t=1}^{T} (U_i(a_i',a_{-i}^t; v_i) -  U_i(a^t;v_i))  = 0.}
\end{equation}
\end{defn}
\begin{defn}[coarse correlated equilibrium, $\CCE$]
A randomized action profile $\al\in \Delta(\A)$ is a {\em coarse
correlated equilibrium} of a complete information game with valuation
profile $\val$ if for every player $i$ and $a_i'\in \A_i$:
\begin{equation}
\E_{\al}\left[U_i(\al; v_i)\right] \geq \E_{\al}\left[U_i(a_i',\al_{-i}; v_i)\right]
\end{equation}
\end{defn}

\begin{theorem}[Blum et al.~\cite{Blum2008}]
\label{thm:no-regret=>CCE} The empirical distribution of actions of any no-regret sequence in a repeated game converges to the set of $\CCE$ of the static game.
\end{theorem}

\paragraph{Price of Anarchy of CCE.}
Roughgarden~\cite{Roughgarden2009} gave a unifying framework for
comparing the social welfare, under various equilibrium notions
including coarse correlated equilibrium, to the optimal social welfare
by defining the notion of a smooth game.  This framework was extended
to games like auctions and allocation mechanisms by
Syrgkanis and Tardos~\cite{Syrgkanis2013}.

\begin{defn}[smooth mechanism]\label{def:smooth-mech}
A mechanism $\M$ is {\em $(\lambda,\mu)$-smooth} for some
$\lambda,\mu\geq 0$ there exists an independent randomized action
profile $\al^*(v) \in \Delta(\A_1)\times\cdots\times\Delta(\A_n)$ for
each valuation profile $v$, such that for any action profile
$a\in \A$ and valuation profile $v \in \V$:
\begin{equation}
\textstyle{\sum_{i\in [n]} U_i(\al_i^*(v),a_{-i}; v_i)\geq \lambda\cdot \opt(v)- \mu\cdot \rev(a)}.
\end{equation}
\end{defn}

Many important games and mechanisms satisfy this smoothness definition
for various parameters of $\lambda$ and $\mu$ (see
Figure~\ref{table:applications}); the following theorem shows that the
welfare of any coarse correlated equilibrium in any of these games is
nearly optimal.

\begin{theorem}[efficiency of CCE; \cite{Syrgkanis2013}]
\label{thm:smooth=>PoA-CCE}
If a mechanism is $(\lambda,\mu)$-smooth then the social welfare of
any course correlated equilibrium at least
$\frac{\lambda}{\max\{1,\mu\}}$ of the optimal welfare, i.e., the
{\em price of anarchy} satisfies $\poa\leq \frac{\max\{1,\mu\}}{\lambda}$.
\end{theorem}

\paragraph{Price of Anarchy of No-regret Learning.}

Following Blum et al.~\cite{Blum2008}, Theorem~\ref{thm:no-regret=>CCE} and
Theorem~\ref{thm:smooth=>PoA-CCE} imply that no-regret learning
dynamics have near-optimal social welfare.

\begin{corollary}[efficiency of no-regret dyhamics; \cite{Syrgkanis2013}]\label{cor:smooth}
If a mechanism is $(\lambda,\mu)$-smooth then the average welfare of
any no-regret dynamics of the repeated game with a fixed player set
and valuation profile, achieves average social welfare at least
$\frac{\lambda}{\max\{1,\mu\}}$ of the optimal welfare, i.e., the
price of anarchy satisfies
$\textstyle{\poa \leq \frac{\max\{1,\mu\}}{\lambda}}$.
\end{corollary}

Importantly, Corollary~\ref{cor:smooth} holds the valuation profile
$\val \in \V$ fixed throughout the repeated game play.  The main
contribution of this paper is in extending this theory to games of
incomplete information, e.g., where the values of the players are
drawn at random in each round of game play.

\begin{figure}
\begin{center}
\begin{tabular}{|c|c|c|c|c|}
        \hline\noalign{\smallskip}
Game/Mechanism  & $(\lambda,\mu)$ & $\poa$ & Reference\\
\hline\hline
\specialcell{Simultaneous First Price Auction with Submodular Bidders} & $(1-1/e,1)$   & $\frac{e}{e-1}$ & \cite{Syrgkanis2013} \\ 
\hline
First Price Multi-Unit Auction &  $(1-1/e,1)$  & $\frac{e}{e-1}$ & \cite{Markakis2013} \\ 
\hline
First Price Position Auction & $(1/2,1)$ & $2$ & \cite{Syrgkanis2013} \\
\hline
All-Pay Auction & $(1/2,1)$ & $2$ & \cite{Syrgkanis2013} \\
\hline
\specialcell{Greedy Combinatorial Auction
with $d$-complements} & $(1-1/e,d)$  & $\frac{d e}{e-1}$ & \cite{Lucier2010} \\
\hline
\specialcell{Proportional Bandwitdth Allocation Mechanism} & $(1/4,1)$  & $4$  & \cite{Syrgkanis2013}\\
\hline
Submodular Welfare Games & $(1,1)$ 
& $2$ & \specialcell{\cite{Vetta2002,Roughgarden2009}} \\
\hline
Congestion Games with Linear Delays & $(5/3,1/3)$ 
   & $5/2$ & \cite{Roughgarden2009}\\
\hline
\end{tabular}
\caption{Examples of smooth games and mechanisms}\label{table:applications} 
\end{center}
\end{figure}

\section{Population Interpretation of Bayesian Games}\label{sec:bce}

In the standard \emph{independent private value model} of a {\em
Bayesian game} there are $n$ players.  Player $i$ has type $\vr_i$
drawn uniformly from the set of type $\V_i$ (and this distribution is
denoted $\F_i$).\footnote{The restriction to the uniform distribution
is without loss of generality for any finite type space and for any
distribution over the type space that involves only rational
probabilities.}  We will restrict attention to the case when the type space $\V_i$ is finite. 
A player's strategy in this Bayesian game is a
mapping $s_i: {\V_i} \to \A_i$ from a valuation $v_i\in \V_i$ to an action
$a_i\in \A_i$.  We will denote with $\Sigma_i=\A_i^{\V_i}$ the
strategy space of each player and with
$\Sigma=\Sigma_1\times\ldots\times \Sigma_n$.  In the game, each
player $i$ realizes his type $v_i$ from the distribution and then
makes action $s_i(v_i)$ in the game.  

In the population interpretation of the Bayesian game, also called
the \emph{agent normal form representation} \cite{Forges1993}, there
are $n$ finite populations of players.  Each player in
population $i$ has a type $v_i$ which we assume to be distinct for
each player in each population and across populations.\footnote{The
restriction to distinct types is without of loss of generality as we
can always augment a type space with an index that does not affect
player utilities.}  The set of players in the population is denoted
$\V_i$.  and the player in population $i$ with type $v_i$ is called
player $v_i$.  In the population game, each player $v_i$ chooses an
action $s_i(v_i)$. Nature uniformly draws one player from each
population, and the game is played with those players' actions.  In
other words, the utility of player $v_i$ from population $i$ is:
\begin{equation}\label{eqn:agent-utilities}
U_{i,v_i}^{\AG}(s)=\E_{\vr}\left[U_i(s(\vr);\vr_i)\cdot 1\{\vr_i = v_i\}\right]
\end{equation}
Notice that the population interpretation of the Bayesian game is in
fact a stochastic game of complete information.

There are multiple generalizations of coarse correlated equilibria
from games of complete information to games of incomplete
information \vsedit{(c.f. \cite{Forges1993}, \cite{Bergemann2011}, \cite{Caragiannis2014})}.
One of the canonical definitions is simply the coarse correlated
equilibrium of the stochastic game of complete information that
is defined by the population interpretation above.\footnote{This
notion is the coarse analog of the \emph{agent normal form Bayes
correlated equilibrium} defined in Section 4.2 of
Forges~\cite{Forges1993}.}

\vsdelete{: the
utility of a player $i$ in the complete information game is his
ex-ante expected utility from the mechanism. The strategy space of
player $i$ is $\Sigma_i=\A_i^{\V_i}$. For a strategy profile
$s\in \Sigma$, the utility of a player in the complete information
game is then:
\begin{equation}
U_i^{\ex}(s) = \E_{\vr}\left[U_i(s(\vr);\vr_i)\right]
\end{equation}
A $\BCCE$ is simply a $\CCE$ of this complete information game. For completeness we provide the formal definition below.}
\vsedit{\begin{defn}[Bayesian coarse correlated equilibrium - $\BCCE$]\label{defn:bcce}
A randomized strategy profile $\sr\in \Delta(\Sigma)$ is a Bayesian coarse correlated equilibrium if for every $a_i'\in A_i$ and for every $v_i\in \V_i$:
\begin{equation}
\E_{\sr}\E_{\vr}\left[U_i(\sr(\vr); \vr_i)~|~\vr_i=v_i\right] \geq \E_{\sr}\E_{\vr}\left[U_i(a_i',\sr_{-i}(\vr_{-i}); \vr_i)~|~\vr_i=v_i\right]
\end{equation}
\end{defn}}

 In a game of  incomplete information the welfare in
equilibrium will be compared to the expected ex-post optimal social
welfare $\E_{\vr}[\opt(\vr)]$.  We will refer to the worst-case ratio
of the expected optimal social welfare over the expected social
welfare of any $\BCCE$ as $\BCCE$-$\poa$.  

\section{Learning in Repeated Bayesian Game}

\vsedit{Consider a repeated version of the population interpretation of a Bayesian game. At each} iteration \vsedit{one player $v_i$}
from each population is sampled uniformly and independently from other
populations. The set of chosen players then participate in an instance
of a mechanism $\M$. We assume that each \vsedit{player $v_i\in \V_i$}, uses
some no-regret learning rule to play in this repeated
game.\footnote{An equivalent and standard way to view a Bayesian game
is that each player draws his value independently from his
distribution each time the game is played.  In this interpretation the
player plays by choosing a strategy that maps his value to an action
(or distribution over actions).  In this interpretation our no-regret
condition requires that the player not regret his actions for each
possible value.}  In Definition~\ref{repeated-game}, we describe the
structure of the game and our notation more elaborately.
\vsedit{\begin{defn}
\label{repeated-game}
The {\em repeated Bayesian game of $\M$} proceeds as
follows. In stage $t$:\vspace{-.5em}
\begin{enumerate}
\item Each player $v_i\in \V_i$ in each population $i$ picks an action $s_{i}^t(v_i)\in A_i$. We denote with $s_i^t\in \A_i^{|\V_i|}$ the function that maps a player $v_i\in \V_i$ to his action.
\item From each population $i$ one player $v_i^t\in \V_i$ is selected uniformly at random. Let $v^t=(v_1^t,\ldots,v_n^t)$ be the chosen profile of players and $s^t(v^t)=(s_1^t(v_1^t),\ldots,s_n^t(v_n^t))$ be the profile of chosen actions.
\item Each player $v_i^t$ participates in an instance of game $\M$, in the role of player $i\in [n]$, with action $s_i^t(v_i^t)$ and experiences a utility of $U_i(s^t(v^t); v_i^t)$. All players  not selected in Step 2 experience zero utility.
\end{enumerate}
\end{defn}\vspace{-1em}
}

\begin{remark}
We point out that for each player in a population to achieve no-regret he does not need to know the distribution of values in other populations. There exist algorithms that can achieve the no-regret property and simply require an oracle that returns the utility of a player at each iteration. Thus all we need to assume is that each player receives as feedback his utility at each iteration.
\end{remark}

\begin{remark}
We also note that our results would extend to the case where at each period multiple matchings are sampled independently and players potentially participate in more than one instance of the mechanism $\M$ and potentially with different players from the remaining population. The only thing that the players need to observe in such a setting is their average utility that resulted from their action $s_i^t(v_i)\in \A_i$ from all the instances that they participated at the given period. Such a scenario seems an appealing model in online ad auction marketplaces where players receive only average utility feedback from their bids.
\end{remark}

\paragraph{Bayesian Price of Anarchy for No-regret Learners.}
In this repeated game setting we want to compare the average social welfare of any sequence of play where each player uses a vanishing regret algorithm versus the average optimal welfare. Moreover, we want to quantify the worst-case such average welfare over all possible valuation distributions within each population:
\begin{equation}
\sup_{\F_1,\ldots,\F_n} \lim\sup_{T\rightarrow \infty} \textstyle{\frac{\sum_{t=1}^{T} \opt(v^t)}{ \sum_{t=1}^T SW^\M(s^t(v^t);v^t)}}
\end{equation}
We will refer to this quantity as the \emph{Bayesian price of anarchy
for no-regret learners}.  The numerator of this term is simply the
average optimal welfare when players from each population are drawn
independently in each stage; it converges almost surely to the
expected ex-post optimal welfare $\E_{\vr}[\opt(\vr)]$ of the stage game.
Our main theorem is that if the mechanism is smooth and players follow
no-regret strategies then the expected welfare is guaranteed to be
close to the optimal welfare.

\begin{theorem}[Main Theorem]
If a mechanism is $(\lambda,\mu)$-smooth then the average \etedit{(over time)} welfare of
any no-regret dynamics of the repeated Bayesian game achieves average
social welfare at least $\frac{\lambda}{\max\{1,\mu\}}$ of the average
optimal welfare, i.e. $\poa\leq \frac{\max\{1,\mu\}}{\lambda}$, \vsedit{almost surely}.
\end{theorem}

\paragraph{Roadmap of the proof.}
\vsdelete{In Section~\ref{sec:bce} we analyze single-shot Bayesian games and define a
Bayesian generalization of coarse correlated equilibria. The repeated
Bayesian game can be viewed as a repetition of a stochastic
game. With this interpretation, i} In
Section \ref{sec:convergence-to-bcce}, we show that any vanishing
regret sequence of play of the repeated Bayesian game, will
converge \emph{almost surely} to the Bayesian version of a coarse
correlated equilibrium of the incomplete information stage
game. Therefore the Bayesian price of total anarchy will be upper
bounded by the efficiency of guarantee of any Bayesian coarse
correlated equilibrium. Finally, in Section \ref{sec:extension} we
show that the price of anarchy bound of smooth mechanisms directly
extends to Bayesian coarse correlated equilibria, thereby providing an
upper bound on the Bayesian price of total anarchy of the repeated
game.

\begin{remark}
We point out that our definition of $\BCCE$ is inherently different and more restricted than the one defined in Caragiannis et al.~\cite{Caragiannis2014}. There, a $\BCCE$ is defined as a joint distribution $D$ over $\V\times \A$, such that if $(\vr,\al)\sim D$ then for any $v_i\in \V_i$ and $a_i'(v_i)\in \A_i$:
\begin{equation}
\textstyle{\E_{(\vr,\al)}\left[U_i(\al; v_i)\right] \geq \E_{(\vr,\al)}\left[U_i(a_i'(\vr_i),\al_{-i}; v_i)\right]}
\end{equation}
The main difference is that the product distribution defined by a distribution in $\Delta(\Sigma)$ and the distribution of values, cannot produce any possible joint distribution over $(\V,\A)$, but the type of joint distributions are restricted to satisfy a conditional independence property described by \cite{Forges1993}. Namely that player $i$'s action is conditionally independent of some other player $j$'s value, given player $i$'s type.
Such a conditional independence property is essential for the guarantees that we will present in this work to extend to a
$\BCCE$ and hence do not seem to extend to the notion given in \cite{Caragiannis2014}.
However, as we will show in Section \ref{sec:convergence-to-bcce}, the no-regret dynamics that we analyze, which are mathematically equivalent to the dynamics in \cite{Caragiannis2014}, do converge to \etedit{this smaller} 
set of $\BCCE$ that we define and for which our efficiency guarantees will extend. This extra convergence property is not needed when the mechanism satisfies the stronger \emph{semi-smoothness} property defined in \cite{Caragiannis2014} and thereby was not needed to show efficiency bounds in their setting.
\end{remark}

\section{Convergence of Bayesian No-Regret to $\BCCE$}\label{sec:convergence-to-bcce}

\vsedit{In this section we show that no-regret learning in the repeated Bayesian game converges almost surely to the set of Bayesian coarse correlated equilibria. Any given sequence of play of the repeated Bayesian game, which we defined in Definition \ref{repeated-game}, gives rise to a sequence of strategy-value pairs $(s^t,v^t)$ where $s^t=(s_1^t,\ldots,s_n^t)$ and $s_i^t\in \A_i^{\V_i}$, captures the actions that each player $v_i$ in population $i$ would have chosen, had they been picked.  Then observe that all that matters to compute the average social welfare of the game for any given time step $T$, is the empirical distribution of pairs $(s,v)$}, up till time step $T$, denoted as $D^T$, i.e. if $(\sr^T,\vr^T)$ is a random sample from $D^T$:
\begin{equation}
\textstyle{\frac{1}{T} \sum_{t=1}^T SW(s^t(v^t); v^t) = \E_{(\sr^T,\vr^T)}\left[ SW(\sr^T(\vr^T); \vr^T)\right]}
\end{equation}

\begin{lemma}[Almost sure convergence to $\BCCE$]\label{lem:convergence-to-product}
\vsedit{Consider a sequence of play of the random matching game, where each player uses a vanishing regret algorithm and let $D^T$ be the empirical distribution of (strategy, valuation) profile pairs up till time step $T$. Consider any subsequence of $\{D^T\}_T$ that converges in distribution to some distribution $D$.} Then, almost surely, $D$ is a product distribution, i.e. $D=D_s\times D_v$, with $D_s \in \Delta(\Sigma)$ and $D_v \times \Delta(\V)$ such that $D_v=\F$ and $D_s\in \BCCE$ of the static incomplete information game with distributional beliefs $\F$.
\end{lemma}
\begin{proof}
\vsedit{We will denote with \[r_i(a_i^*, a; v_i)=U_i(a_i^*,a_{-i}; v_i) - U_i(a;v_i),\] the regret of player $v_i$ from population $i$, for action $a_i^*$ at action profile $a$. For a $v_i\in \V_i$ let $x_i^t(v_i)={\bf 1}\{v_i^t=v_i\}$. Since the sequence has vanishing regret for each player $v_i$ in population $P_i$, it must be that for any $s_i^*\in \Sigma_i$:
\begin{equation}\label{eqn:value-regret}
\textstyle{\sum_{t=1}^{T} x_i^t(v_i)\cdot r_i\left( s_i^*(v_i), s^t(v^t); v_i\right) \leq  o(T)}
\end{equation}}

\vsedit{For any fixed $T$, let $D_s^T\in \Delta(\Sigma)$ denote the empirical distribution of $s^t$ and let $\sr$ be a random sample from $D_s^T$. For each $s\in \Sigma$, let $\T_{s}\subset [T]$ denote the time steps such that $s^t=s$ for each $t\in \T_{s}$. Then we can re-write Equation \eqref{eqn:value-regret} as:
\begin{equation}\label{eqn:pre-pre-limit}
\textstyle{\E_{\sr}\left[ \frac{1}{|\T_{\sr}|} \sum_{t\in \T_{\sr}} x_i^t(v_i)\cdot r_i\left( s_i^*(v_i), s^t(v^t); v_i\right) \right]\leq \frac{o(T)}{T}}
\end{equation}
For any $s\in \Sigma$ and $w\in \V$, let $\T_{s,w}=\{t\in \T_{s}: v^t=w\}$. Then we can re-write Equation \eqref{eqn:pre-pre-limit} as:
\begin{equation}\label{eqn:pre-limit}
\textstyle{\E_{\sr}\left[ \sum_{w\in \V}\frac{|\T_{\sr,w}|}{|\T_{\sr}|} 1\{w_i=v_i\}\cdot r_i\left(s_i^*(v_i), \sr(w); v_i\right)\right]\leq   \frac{o(T)}{T}}
\end{equation}}

Now we observe that $\frac{|\T_{s,w}|}{|\T_{s}|}$ is the empirical frequency of the valuation vector $w\in \V$, when filtered at time steps where the strategy vector was $s$.  Since at each time step $t$ the valuation vector $v^t$ is picked independently from the distribution of valuation profiles $\F$, this is the empirical frequency of $\T_{s}$ independent samples from $\F$. 

By standard arguments from empirical processes theory, if $\T_{s}\rightarrow \infty$ then this empirical distribution converges almost surely to the distribution $\F$. On the other hand if $\T_{s}$ doesn't go to $\infty$, then the empirical frequency of strategy $s$ vanishes to $0$ as $T\rightarrow \infty$ and therefore has measure zero in the above expectation as $T\rightarrow \infty$. Thus for any convergent subsequence of $\{D^T\}$, if $D$ is the limit distribution, then if $s$ is
in the support of $D$, then almost surely the distribution of $w$ conditional on strategy $s$ is $\F$. Thus we can write $D$ as a product distribution $D_s\times \F$. 

\vsedit{Moreover, if we denote with ${\bf w}$ the random variable that follows distribution $\F$, then the limit of Equation \eqref{eqn:pre-limit} for any convergent sub-sequence, will give that:
\begin{equation*}
\text{a.s.: }\E_{\sr\sim D_s} \E_{{\bf w}\sim \F}\left[ 1\{{\bf w}_i=v_i\}\cdot r_i\left(s_i^*(v_i),\sr({\bf w}); v_i\right)\right] \leq 0
\end{equation*}
Equivalently, we get that $D_s$ will satisfy that for all $v_i\in \V_i$ and for all $s_i^*$:
\begin{equation*}
\text{a.s.: }\E_{\sr\sim D_s} \E_{{\bf w}\sim \F}\left[r_i\left(s_i^*({\bf w}_i),\sr({\bf w}); {\bf w}_i\right)~|~ {\bf w}_i = v_i\right] \leq 0
\end{equation*}
The latter is exactly the $\BCCE$ condition from Definition \ref{defn:bcce}. Thus $D_s$ is in the set of $\BCCE$ of the static incomplete incomplete information game among $n$ players, where the type profile is drawn from $\F$.}
\end{proof}

\vsedit{Given the latter convergence theorem we can easily conclude the following the following theorem, whose proof is given in the supplementary material.
\begin{theorem}\label{thm:bce_poa_pota}
The price of anarchy for Bayesian no-regret dynamics is upper bounded
by the price of anarchy of Bayesian coarse correlated equilibria, almost surely.
\end{theorem}
}

\section{Efficiency of Smooth Mechanisms at Bayes Coarse Correlated Equilibria}\label{sec:extension}

In this section we show that smoothness of a mechanism $\M$ implies that any $\BCCE$ of the incomplete information setting achieves at least $\frac{\lambda}{\max\{1,\mu\}}$ of the expected optimal welfare. \vsedit{To show this we will adopt the interpretation of $\BCCE$ that we used in the previous section, as coarse correlated equilibria of a more complex normal form game; the stochastic agent normal form representation of the Bayesian game. We can interpret this complex normal form game as the game that arises from a complete information mechanism $\M^{\AG}$ among $\sum_i |\V_i|$ players, which randomly samples one player from each of the $n$ population and where the utility of a player in the complete information mechanism $\M^{\AG}$ is given by Equation \eqref{eqn:agent-utilities}. The set of possible outcomes in this agent game corresponds to the set of mappings from a profile of chosen players to an outcome in the underlying mechanism $\M$. The optimal welfare of this game, is then the expected ex-post optimal welfare $\opt^{\AG} = \E_{\vr}\left[\opt(\vr)\right]$.}

The main theorem that we will show is that whenever mechanism $\M$ is $(\lambda,\mu)$-smooth, then also mechanism $\M^{\AG}$ is $(\lambda,\mu)$-smooth. Then we will invoke a theorem of \cite{Syrgkanis2013,Roughgarden2009}, which shows that any coarse correlated equilibrium of a complete information mechanism achieves at least $\frac{\lambda}{\max\{1,\mu\}}$ of the optimal welfare. \vsedit{By the equivalence between $\BCCE$ and $\CCE$ of this complete information game, we get that every $\BCCE$ of the Bayesian game achieves at least $\frac{\lambda}{\max\{1,\mu\}}$ of the expected optimal welfare. }

\vsdelete{More formally, we define a mechanism $\M^{\AG}$
 where the action space of a player is the set of functions from a type to an action of the complete information setting, $\A_i^{\ex}\triangleq (\V_i \rightarrow \A_i)$. 
The utility of a player in this mechanism is his ex-ante expected utility, $U_i^{\ex}(s) = \E_{\vr}\left[U_i^\M(s(v))\right]$. Observe that the utility $U_i^{\ex}$ does not depend on any private information of the player. It only depends on the underlying utility function $U_i^\M$ and on the distributions $\F_i$ which are all common knowledge. The set of possible outcomes in this ex-ante game corresponds to the set of mappings from an allocation profile to an outcome in the underlying mechanism $\M$. The optimal welfare of this game, is then the expected ex-post optimal welfare $\opt^{\ex} = \E_{\vr}\left[\opt(\vr)\right]$.
The main result of this section is to show that if a mechanism $\M$ is smooth according to definition
\ref{def:smooth-mech} then $\M^{\ex}$ is also a smooth mechanism and therefore by \cite{Syrgkanis2013,Roughgarden2009} every coarse correlated equilibria of $\M^{\ex}$, which corresponds to a $\BCCE$ of the incomplete information game, will achieve at least $\frac{\lambda}{\max\{1,\mu\}}$ of the expected optimal welfare. }

\begin{theorem}[From complete information to Bayesian smoothness]\label{thm:extension-theorem}
If a mechanism $\M$ is $(\lambda,\mu)$-smooth, then for any vector of independent valuation distributions $\F=(\F_1,\ldots,\F_n)$, the complete information mechanism $\M^{\AG}$ is also $(\lambda,\mu)$-smooth.
\end{theorem}
\begin{proof}
\vsedit{Consider the following randomized deviation for each player $v_i\in\V_i$ in population $i$}: \vsdelete{that depends only on the information that he has which is his own value $v_i$:} He random samples a valuation profile $\wb\sim \F$. Then he plays according to the randomized action $\sr_i^*(v_i,\wb_{-i})$, i.e., the player deviates using the randomized action guaranteed by the smoothness property of mechanism $\M$ for his type $v_i$ and the random sample of the types of the others $\wb_{-i}$.

\vsedit{Consider an arbitrary action profile $s=(s_1,\ldots,s_n)$ for all players in all populations. In this context it is better to think of each $s_i$ as a $|\V_i|$ dimensional vector in $\A_i^{|\V_i|}$ and to view $s$ as a $\sum_i |\V_i|$ dimensional vector. Then with $s_{-v_i}$ we will denote all the components of this large vector except the ones corresponding to player $v_i\in\V_i$. Moreover, we will be denoting with $\vr$ a sample from $\F$ drawn by mechanism $\M^{\AG}$. We now argue about the expected utility of player $v_i$ from this deviation, which is:
\begin{align*}
\E_{\wb}\left[U_{i,v_i}^{\AG}(s_i^*(v_i,\wb_{-i}),s_{-{v_i}})\right]=~&\E_{\wb}\E_{\vr}\left[U_i(s_i^*(v_i,\wb_{-i}),s_{-i}(\vr_{-i});v_i)\cdot 1\{\vr_i = v_i\}\right]
\end{align*}
Summing the latter over all players $v_i\in \V_i$ in population $i$:
\begin{align*}
\sum_{v_i\in \V_i} \E_{\wb}\left[U_{i,v_i}^{\AG}(s_i^*(v_i,\wb_{-i}),s_{-{v_i}})\right]
=~&\textstyle{\E_{\wb,\vr}\left[\sum_{v_i\in \V_i}U_i(s_i^*(v_i,\wb_{-i}),s_{-i}(\vr_{-i});v_i)\cdot 1\{\vr_i = v_i\}\right]}\\
=~& \E_{\vr,\wb} \left[U_i(\sr_i^*(\vr_i,\wb_{-i}),s_{-i}(\vr_{-i}); \vr_i)\right]\\
=~& \E_{\vr,\wb} \left[U_i(\sr_i^*(\wb_i,\wb_{-i}),s_{-i}(\vr_{-i}); \wb_i)\right]\\
=~& \E_{\vr,\wb} \left[U_i(\sr_i^*(\wb),s_{-i}(\vr_{-i}); \wb_i)\right],
\end{align*}
where the second to last equation is an exchange of variable names and regrouping using independence. Summing over populations and using smoothness of $\M$, we get smoothness of $\M^{\AG}$:
\begin{align*}
\sum_{i\in [n]}\sum_{v_i\in \V_i} &\E_{\wb}\left[U_{i,v_i}^{\AG}(s_i^*(v_i,\wb_{-i}),s_{-{v_i}})\right] =~\textstyle{\E_{\vr,\wb} \left[\sum_{i\in [n]} U_i(\sr_i^*(\wb),s_{-i}(\vr_{-i}); \wb_i)\right]}\\
\geq~& \E_{\vr,\wb}\left[\lambda \opt(\wb) - \mu\rev(s(\vr))\right]
=~ \lambda \E_{\wb}\left[\opt(\wb)\right] - \mu \rev^{AG}(s)\end{align*}
}
\end{proof}

\begin{corollary}\label{cor:efficiency-of-bcce} Every $\BCCE$ of the incomplete information setting of a smooth mechanism $\M$, achieves expected welfare at least $\frac{\lambda}{\max\{1,\mu\}}$ of the expected optimal welfare.
\end{corollary} 
\section{Finite Time Analysis and Convergence Rates}

In the previous section we argued about the limit average efficiency of the game as time goes to infinity. In this section we analyze the convergence rate to $\BCCE$ and we show approximate efficiency results even for finite time, when players are allowed to have some $\epsilon$-regret. 

\begin{theorem}\label{thm:finite_conv}
Consider the repeated matching game with a $(\lambda,\mu)$-smooth mechanism. Suppose that for any $T\geq T^0$, each player in each of the $n$ populations has regret at most $\frac{\epsilon}{n}$. Then for every $\delta$ and $\rho$, there exists a $T^*(\delta,\rho)$, such that for any $T\geq \min\{T^0,T^*\}$, with probability $1-\rho$:
\begin{equation}
\textstyle{\frac{1}{T} \sum_{t=1}^{T} SW(s^t(v^t);v^t) \geq \frac{\lambda}{\max\{1,\mu\}}\E_{\vr}\left[\opt(\vr)\right] - \delta - \mu\cdot \epsilon}
\end{equation}
Moreover, $T^*(\delta,\rho) \leq \frac{54\cdot n^3\cdot |\Sigma| \cdot |\V|^2 \cdot H^3}{\delta^3} \log\left(\frac{2}{\rho}\right)$.
\end{theorem}

\bibliographystyle{plain}
\bibliography{thesis-bib}
\newpage

\begin{appendix}

\begin{center}
\bf \Large Supplementary material for \\ ``No-Regret Learning in Bayesian Games''
\end{center}
\setcounter{page}{1}


\section{Proof of Theorem \ref{thm:bce_poa_pota}}

For readability we repeat the definitions of Lemma \ref{lem:convergence-to-product} and Theorem \ref{thm:bce_poa_pota} from the main text. 

\begin{rtheorem}{Lemma}{\ref{lem:convergence-to-product}}
Let $D\in \Delta(\Sigma\times \V)$ be a joint distribution of (strategy, valuation) profile pairs. Consider a sequence of play of the random matching game, where each player uses a vanishing regret algorithm and let $D^T$ be the empirical distribution of strategy, valuation profile pairs up till time step $T$. Suppose that there exists a subsequence of $\{D^T\}_T$ that converges in distribution to $D$. Then, almost surely, $D$ is a product distribution, i.e. $D=D_s\times D_v$, with $D_s \in \Delta(\Sigma)$ and $D_v \times \Delta(\V)$ such that $D_v=\F$ and $D_s\in \BCCE$ of the static incomplete information game with distributional beliefs $\F$.
\end{rtheorem}

\begin{rtheorem}{Theorem}{\ref{thm:bce_poa_pota}}
The price of anarchy for Bayesian no-regret dynamics is upper bounded
by the price of anarchy of Bayesian coarse correlated equilibria.
\end{rtheorem}
\begin{proof}
Let $D\in \Delta(\Sigma\times \V)$ be a joint distribution, such that there is a subsequence of $\{D^T\}_T$, converging in distribution to $D$. Then by Lemma \ref{lem:convergence-to-product}, almost surely, $D$ is a product distribution, i.e. $D\in \Delta(\Sigma)\times \Delta(\V)$ and that the marginal on $\V$ is equal to $\F$ and the marginal on $\Sigma$ is a $\BCCE$ of the static incomplete information game with distributional beliefs $\F$.

Therefore, if $\rho$ is the $\BCCE-\poa$ of the mechanism, and if $(\sr,\vr)$ is a random sample from $D$, then almost surely:
\begin{equation}
\E_{\sr,\vr}\left[SW(\sr(\vr);\vr)\right] \geq \frac{1}{\rho}\E_{\vr}\left[\opt(\vr)\right]
\end{equation}
Thus the limit average social welfare of any convergent subsequence will be at least $\frac{1}{\rho}\E_{\vr}\left[\opt(\vr)\right]$, which then implies that almost surely:
\begin{equation*}
\lim\inf_{T\rightarrow \infty} \frac{1}{T} \sum_{t=1}^T SW(s^t(v^t); v^t)\geq \frac{1}{\rho}\E_{\vr}\left[\opt(\vr)\right] = \frac{1}{\rho}\lim_{T\rightarrow \infty} \frac{1}{T}\sum_{t=1}^{T} \opt(v^t)
\end{equation*}
Thus for any non-measure zero event, for any $\epsilon$, there exists a $f(\epsilon)$ such that for any $T\geq f(\epsilon)$:
\begin{equation*}
\frac{1}{T} \sum_{t=1}^T SW(s^t(v^t); v^t)\geq \frac{1}{\rho}\frac{1}{T}\sum_{t=1}^{T} \opt(v^t) - \epsilon
\end{equation*}
With no loss of generality we can assume that $\E_{\vr}\left[\opt(\vr)\right]>0$ (o.w. valuations are all zero and theorem holds trivially). Since, the average optimal welfare converges almost surely to $\E_{\vr}\left[\opt(\vr)\right]$, we get that for any non-measure zero event, there exists a $g(\delta)$ such that for
$T\geq g(\delta)$, $\frac{1}{T}\sum_{t=1}^{T} \opt(v^t)$ is bounded away from zero. Thereby, we can turn the additive error into a multiplicative one, i.e. for any non-measure zero event and for any $\epsilon'$ there exists $w(\epsilon')$ such that for any $T\geq w(\epsilon')$:
\begin{equation*}
\frac{1}{T} \sum_{t=1}^T SW(s^t(v^t);v^t)\geq \frac{1}{\rho}\left(1+\epsilon'\right)\frac{1}{T}\sum_{t=1}^{T} \opt(v^t) 
\end{equation*}
This implies that almost surely:
\begin{equation*}
\lim\sup_{T\rightarrow \infty} \frac{\frac{1}{T}\sum_{t=1}^{T} \opt(v^t)}{\frac{1}{T} \sum_{t=1}^T SW(s^t(v^t); v^t)}\leq \rho = \BCCE\text{-}\poa
\end{equation*}
\end{proof}

\section{Proof of Theorem \ref{thm:finite_conv}}

\begin{rtheorem}{Theorem}{\ref{thm:finite_conv}}
Consider the repeated matching game with a $(\lambda,\mu)$-smooth mechanism. Suppose that for any $T\geq T^0$, each player in each of the $n$ populations has regret at most $\frac{\epsilon}{n}$. Then for every $\delta$ and $\rho$, there exists a $T^*(\delta,\rho)$, such that for any $T\geq \min\{T^0,T^*\}$, with probability $1-\rho$:
\begin{equation}
\frac{1}{T} \sum_{t=1}^{T} SW(s^t(v^t);v^t) \geq \frac{\lambda}{\max\{1,\mu\}}\E_{\vr}\left[\opt(\vr)\right] - \delta - \mu\cdot \epsilon
\end{equation}
Moreover, $T^*(\delta,\rho) \leq \frac{54\cdot n^3\cdot |\Sigma| \cdot |\V|^2 \cdot H^3}{\delta^3} \log\left(\frac{2}{\rho}\right)$.
\end{rtheorem}
\begin{proof}
Fix a population $i$ and a Bayesian strategy $s_i^* \in \Sigma_i$, as well as a Bayesian strategy profile $s\in \Sigma$. For shorter notation we will denote: \[\pi_i(s_i^*,s,v) = U_i(s_i^*(v_i),s_{-i}(v_{-i}); v_i).\] For a time step $T$, let $p^T(s) = \frac{|\T_s|}{T}$ be the empirical distribution of a Bayesian strategy $s$ and with $p^T(v|s) = \frac{|\T_{s,v}|}{|\T_s|}$ be the empirical distribution of values conditional on a Bayesian strategy $s$. The average utility of a population $i$ up till time step $T$, when switching to a fixed Bayesian strategy $s_i^*$, can be written as:
\begin{equation}
\frac{1}{T} \sum_{t=1}^{T} \pi(s_i^*, s^t, v^t) = \sum_{s\in \Sigma} p^T(s) \sum_{v\in \V} p^T(v|s) \cdot\pi_i(s_i^*,s,v)
\end{equation}

We will show that for any $s_i^*$, there exists a $T^*(\delta,\rho)$ such that for any $T\geq T^*(\delta,\rho)$, with probability $1-\rho$:
\begin{equation}\label{eqn:closeness-of-utilities-with-empirical}
 \sum_{s\in \Sigma} p^T(s) \sum_{v\in \V} p^T(v|s) \cdot\pi_i(s_i^*,s,v) \geq \sum_{s\in \Sigma} p^T(s) \E_{\vr}\left[ \pi_i(s_i^*,s,\vr)\right] - \delta
\end{equation}
where $\vr$ is a random variable drawn from the distribution of valuation profiles $\F$. We will denote with $p(v)$ the density function implied by distribution $\F$.

In what follows we will denote with $H=\max_{i\in [n], v_i \in \V_i, x_i\in \X_i} v_i(x_i)$ the maximum possible value of any player. Thus observe that the utility of any player is upper bounded by $H$ and that the revenue collected by any player at equilibrium is upper bounded by $H$. 

For a time period $T$, let $G=\{s\in \Sigma: p^T(s) \geq \zeta\}$. Then observe that:
\begin{multline*}
 \sum_{s\in \Sigma} p^T(s) \sum_{v\in \V} \left(p^T(v|s)-p(v)\right) \cdot\pi_i(s_i^*,s,v)\geq\\ \sum_{s\in G} p^T(s) \sum_{v\in V} \left(p^T(v|s)-p(v)\right) \cdot\pi_i(s_i^*,s,v) - \zeta \cdot |\Sigma|\cdot H
\end{multline*}
Observe that for any $s\in G$, $|\T_s|\geq \zeta \cdot T$. Thus $p^T(v|s)$ is the empirical mean of at least $\zeta\cdot T$ independent random samples of a Bernoulli trial with success probability $p(v)$. Hence, by Hoeffding bounds, we have that $|p^T(v|s)-p(v)|\leq t$ with probability at least $1-2\exp\left(-2\cdot \zeta\cdot T\cdot t^2\right)$. Thus with that much probability we get:
\begin{align*}
 \sum_{s\in \Sigma} p^T(s) \sum_{v\in \V} \left(p^T(v|s)-p(v)\right) \cdot\pi_i(s_i^*,s,v)\geq -t\cdot |\V|\cdot H - \zeta \cdot |\Sigma|\cdot H
\end{align*}
By setting $t=\frac{\delta}{2\cdot |\V|\cdot H}$, $\zeta = \frac{\delta}{2\cdot |\Sigma| \cdot H}$ and $T^*(\delta,\rho)= \frac{16 \cdot |\Sigma| \cdot |\V|^2 \cdot H^3}{\delta^3} \log\left(\frac{2}{\rho}\right)$, we get the claimed property in Equation \eqref{eqn:closeness-of-utilities-with-empirical}.

Now suppose that after time step $T^0$ each player in a population has regret $\epsilon/n$. Thus the average utility of the population is at least the utility from switching to any fixed Bayesian strategy $s_i^*$, minus an error term of $\epsilon/n$:
\begin{equation}
\sum_{s\in \Sigma} p^T(s) \sum_{v\in \V} p^T(v|s) \pi_i(s_i,s,v)\geq \sum_{s\in \Sigma} p^T(s) \sum_{v\in \V} p^T(v|s) \pi_i(s_i^*,s,v)- \frac{\epsilon}{n}
\end{equation}
From the previous analysis, for any $T\geq \min\{T^0,T^*(\frac{2\delta}{3\cdot n},\rho)\}$, we get that with probability $1-\rho$:
\begin{equation}
\sum_{s\in \Sigma} p^T(s) \sum_{v\in \V} p^T(v|s) \pi_i(s_i,s,v)\geq \sum_{s\in \Sigma} p^T(s) \E_{\vr}\left[\pi_i(s_i^*,s,\vr)\right]- \frac{2\delta}{3n}- \frac{\epsilon}{n}
\end{equation}
Summing over all populations and using the Bayesian smoothness property of the mechanism from Theorem \ref{thm:extension-theorem}, we have that with probability $1-\rho$:
\begin{align*}
\sum_{s\in \Sigma} p^T(s) \sum_{v\in \V} p^T(v|s) \sum_i \pi_i(s_i,s,v)\geq~& \sum_{s\in \Sigma} p^T(s) \left(\lambda \E_{\vr}\left[\opt(\vr)\right] - \mu \rev^{\ex}(s)\right)- \frac{2\delta}{3}- \epsilon\\
\geq~& \lambda \E_{\vr}\left[\opt(\vr)\right] - \mu \sum_{s\in \Sigma} p^T(s) \rev^{\AG}(s) - \frac{2\delta}{3} - \epsilon
\end{align*}

To conclude the theorem we observe that since for any $s\in \Sigma$, $|p^T(v|s)-p(v)|\leq \frac{\delta}{3 \cdot n\cdot |\V|\cdot H}$, we get that:
\begin{equation}
\rev^{\ex}(s) = \sum_{v\in \V} p(v) \rev(s(v)) \leq \sum_{v\in \V} p^T(v|s) \rev(s(v)) + \frac{\delta}{3}
\end{equation}
Since, the revenue collected by a player at any action in the support of an equilibrium is at most $H$.  By the latter we can combine the revenue on the right hand side with the utility on the left hand side. We can also bound the remaining $(\mu-1)$ of the revenue, by $(\mu-1)$ of the average welfare minus $\epsilon$, since each player in each population can always drop out of the auction and therefore his average utility at an $\frac{\epsilon}{n}$-regret sequence must be at least $-\frac{\epsilon}{n}$. 

Hence, we get that:
\begin{equation}
\sum_{s\in \Sigma} p^T(s) \sum_{v\in \V} p^T(v|s) SW(s(v);v) \geq \frac{\lambda}{\max\{1,\mu\}} \E_{\vr}\left[\opt(\vr)\right] - \delta - \mu \cdot \epsilon
\end{equation}
Thus choosing $T^*(\rho,\frac{2\delta}{3\cdot n}) = \frac{54\cdot n^3\cdot |\Sigma| \cdot |\V|^2 \cdot H^3}{\delta^3} \log\left(\frac{2}{\rho}\right)$, we get the conditions of the theorem.
\end{proof}
%

\end{appendix}
\end{document}